\RequirePackage{luatex85}

\documentclass[11pt,letterpaper]{article}

\usepackage{iftex}

\ifPDFTeX
\usepackage[utf8]{inputenc}
\usepackage[T1]{fontenc}
\fi

\usepackage{amsmath}
\usepackage{amssymb}
\usepackage{amsthm}
\usepackage{thmtools}

\ifPDFTeX

\usepackage{calrsfs}
\DeclareMathAlphabet{\pazocal}{OMS}{zplm}{m}{n}

\usepackage{libertine}
\usepackage[scaled=0.96]{zi4} % use Inconsolata as the typewriter font
\else
\usepackage{fontspec}
\usepackage{unicode-math}
\fi

\usepackage[paperwidth=8.5in,paperheight=11in,left=1in,top=1in,right=1in,bottom=1in]{geometry}

\declaretheorem[numberwithin=section,refname={Theorem,Theorems},Refname={Theorem,Theorems},name={Theorem}]{theorem}
\declaretheorem[numberlike=theorem,refname={Lemma,Lemmas},Refname={Lemma,Lemmas},name={Lemma}]{lemma}

\declaretheorem[numberlike=theorem,refname={Definition,Definitions},Refname={Definition,Definitions},name={Definition}]{definition}

\usepackage{enumitem}
\usepackage{hyperref}
\usepackage{graphicx}
\usepackage[ruled,vlined,linesnumbered]{algorithm2e}
\usepackage[utf8]{inputenc} 
\usepackage{amsmath,amssymb,amsthm} 
\usepackage{url}
\usepackage{booktabs}
\usepackage{todonotes}

\newcommand{\dmax}{\mathit{d_{\textrm{max}}}}
\newcommand{\dmin}{\mathit{d_{\textrm{min}}}}

\newcommand{\Z}{\mathbb{Z}}

\newcommand{\opt}{\textsc{OPT}}

\newcommand{\poly}{\textrm{poly}}

\newcommand{\dd}{\kappa}
\newcommand{\ball}[2]{B\!\left(#1,#2\right)}

\newcommand{\ms}{M}

\newcommand{\OPT}{\textsc{OPT}}
\newcommand{\cost}{\mathit{cost}}

\newcommand{\eps}{\epsilon}
\newcommand{\ar}{\Delta}
\newcommand{\hierarchy}{\Gamma}

\newcommand{\rpmax}{r^p_{\max}}
\newcommand{\rpmin}{r^p_{\min}}

\newcommand{\AChan}{\mathcal{A}_\mathrm{CGS}}
\newcommand{\ACov}{\mathcal{A}_\mathrm{Cov}}

\newcommand{\logDloglogD}{\log \Delta \log\log\Delta}

\usepackage[libertine]{newtxmath} % must load after ams packages
\usepackage{cleveref}

\title{Fully Dynamic \texorpdfstring{$k$-Center}{k-Center} Clustering in Doubling Metrics\thanks{The research leading to these results has received funding from the European Research Council under the European Union's Seventh Framework Programme (FP/2007-2013) / ERC Grant Agreement no. 340506.}}

\author{
	Gramoz Goranci\thanks{
		University of Toronto, Canada} 
	\and 
	Monika Henzinger\thanks{
		University of Vienna, Austria}
	\and
	Dariusz Leniowski \thanks{
		University of Vienna, Austria}
	\and
	Christian Schulz \thanks{
		University of Vienna, Austria}
	\and \\[0.05cm]
	Alexander Svozil\thanks{
		University of Vienna, Austria}	
}

\date{}

\begin{document}

	\maketitle

	\begin{abstract}
	
Clustering is one of the most fundamental problems in unsupervised learning with a large number of
		applications. However, classical clustering algorithms assume that the data is static, thus failing
		to capture many real-world applications where data is constantly changing and evolving. Driven by
		this, we study the metric $k$-center clustering problem in the fully dynamic setting, where the goal
		is to efficiently maintain a clustering while supporting an intermixed sequence of insertions and
		deletions of points. This model also supports queries of the form (1) report whether a given point
		is a center or (2) determine the cluster a point is assigned to. 
		
		We present a deterministic dynamic
		algorithm for the $k$-center clustering problem that provably achieves a $(2+\epsilon)$-approximation
		in poly-logarithmic update and query time, if the underlying metric has bounded doubling dimension,
		its aspect ratio is bounded by a polynomial and $\epsilon$ is a constant. An important feature of
		our algorithm is that the update and query times are independent of $k$. We confirm the practical
		relevance of this feature via an extensive experimental study which shows that for 
		values of $k$ and $\epsilon$ suggested by theory, 
		our algorithmic construction outperforms the state-of-the-art
		algorithm in terms of solution quality and running time.

	\end{abstract}
	\newpage
	\section{Introduction}\label{sec:intro}
	The massive increase in the amount of data produced over the last few decades has motivated the study of different tools for analysing and computing specific properties of the data. One of the most extensively studied analytical tool is clustering, where the goal is to group the data into clusters of ``close'' data points. Clustering is a fundamental problem in computer science and it has found a wide range of applications in unsupervised learning, classification, community detection, image segmentation and databases (see e.g.~\cite{FORTUNATO201075,Schaeffer07,ShiM00}). 
	
	A natural definition of clustering is the \emph{$k$-center clustering}, where given a set of $n$
	points in a metric space and a parameter $k \leq n$, the goal is to select  $k$ designated
	points, referred to as \emph{centers}, such that their \emph{cost}, defined as the maximum
	distance of any point to its closest center, is minimized. As finding the optimal $k$-center
	clustering is NP-hard~\cite{KH79}, the focus has been on studying the approximate version of
	this problem. For a parameter $\alpha \geq 1$, an $\alpha$-approximation to the $k$-center
	clustering problem is an algorithm that outputs $k$ centers such that their cost is within
	$\alpha$ times the cost of the optimal solution. There is a simple $2$-approximate $k$-center
	clustering algorithm by Gonzalez~\cite{G85} that runs in $O(nk)$ time; repeatedly pick the point furthest away from the current set of centers as the next center to be added. The problem of finding a $(2-\epsilon)$-approximate $k$-center clustering is known to be NP-complete~\cite{G85}.

	In many real-world applications, including social networks and the Internet, the data is subject to frequent updates over time. For example, every second about thousands of Google searches, YouTube video uploads and Twitter posts are generated. However, most of the traditional clustering algorithms are not capable of capturing the dynamic nature of data and often completely reclustering from scratch is used to obtain desirable clustering guarantees. 
	
	To address the above challenges, in this paper we study a \emph{dynamic} variant of the
	$k$-center clustering problem, where the goal is to maintain a clustering with small
	approximation ratio while supporting an intermixed update sequence of insertions and deletions
	of points with a small time per update. Additionally, for any given point we want to report
	whether this point is a center or determine the cluster this point is assigned to. When only
	insertions of points are allowed, also known as the \emph{incremental} setting, Charikar et al.~\cite{CharikarCFM04} designed an $8$-approximation algorithm with $O(k \log k)$ amortized
	time per point insertion. This result was later improved to a $(2 + \epsilon)$-approximation by McCutchen and Khuller~\cite{MK08}.  Recently, Chan et al.~\cite{ChanGS18} studied the model
	that supports both point insertions and deletions, referred to as the \emph{fully-dynamic}
	setting. Their dynamic algorithm is randomized and achieves a $(2+\epsilon)$-approximation with
	$O(k^{2} \cdot \epsilon^{-1} \cdot \log \Delta)$ update time per operation, where $\Delta$ is
	the aspect ratio of the underlying metric space.   %Moreover, their running time guarantees hold against an oblivious adversary, i.e., the adversary does not have access to the random choices made by the algorithm.
	
	%The ``curse of dimensionality'' is a well-known phenomena in data analysis and it generally occurs when the metric defined on the data is high dimensional. The metrics that suffers from the increase in the dimensionality are often very complex to deal with since performing computational task on the data often requires very high time and space complexity. One well-studied approach to address this problem is to define a notion of dimension on finite metric spaces, and develop efficient algorithms that are specifically tailored to this particular notion. Arguably one the most common concepts is doubling dimension. 
	
It is an open question whether there are fully-dynamic algorithms that achieve smaller running time (ideally independent of $k$) while still keeping the same approximation guarantee. We study such data structures for metrics spaces with ``limited expansion''. More specifically we consider the well-studied notion of doubling dimension. The \emph{doubling dimension} of a metric space is bounded by $\kappa$ if any ball of radius $r$ in this metric can be covered by $2^{\kappa}$ balls of radius $r/2$~\cite{KrauthgamerL04}. This notion can be thought of as a generalization of the Euclidean dimension since $\mathbb{R}^{d}$ has doubling dimension $\Theta(d)$.
	
	The $k$-center clustering problem has been studied in the low dimensional regime from both the
	static and dynamic perspective. Feder and Greene~\cite{FederG88} showed that if the input points
	are taken from $\mathbb{R}^{d}$, there is a $2$-approximation to the optimal clustering that can
	be implemented in $O(n \log k)$ time. They also showed that computing an approximation better
	than $1.732$ is NP-hard, even when restricted to Euclidean spaces. For metrics of bounded
	doubling dimension Har{-}Peled and Mendel~\cite{Har-PeledM06} devised an algorithm that achieves a $2$-approximation
	and runs in $O(n \log n)$ time. In the dynamic setting, Har{-}Peled~\cite{Har-Peled04} implicitly gave a fully-dynamic algorithm for metrics with bounded doubling dimension that reports a $(2+\epsilon)$-clustering at any time while supporting insertion and deletions of points in $O(\poly(k,\epsilon^{-1}, \log n))$ time, where $\poly(\cdot)$ is a fixed-degree polynomial in the input parameters.
	
	One drawback shared by the above dynamic algorithms for the $k$-center clustering is that the
	update time is dependent on the number of centers $k$. This is particularly undesirable in the
	applications where $k$ is relatively large. For example, one application where this is justified
	is the distribution of servers on the Internet, where thousands of servers are heading towards
	millions of routers. Moreover, this dependency on $k$ seems inherent in the state-of-the-art
	dynamic algorithms; for example, the algorithm due to Chan et al.~\cite{ChanGS18} requires
	examining the set of current centers upon insertion of a point, while the algorithm due
	to Har{-}Peled~\cite{Har-Peled04} employs the notion of \emph{coresets}, which in turn requires dependency on the number of centers.  
	
%<<<<<<< HEAD
	%In this paper we present a dynamic algorithm for metrics with bounded doubling dimension that achieves a $(2+\epsilon)$ approximation ratio for the $k$-center clustering problem~(thus matching the approximation ratio of the dynamic algorithm in general metric spaces~\cite{ChanGS18}) while supporting insertions and deletion of points in time \emph{independent} of the number of centers $k$ and poly-logarithmic in the aspect ratio $\Delta$. Our algorithm is  deterministic and thus works against an adaptive adversary. %Our first and main result is an incremental algorithm and we summarize its guarantees in the theorem below.

%=======
	In this paper we present a dynamic algorithm for metrics with bounded doubling dimension that achieves a $(2+\epsilon)$ approximation ratio for the $k$-center clustering problem~(thus matching the approximation ratio of the dynamic algorithm in general metric spaces~\cite{ChanGS18}) while supporting insertions and deletion of points in time \emph{independent} of the number of centers $k$ and poly-logarithmic in the aspect ratio $\Delta$. Our algorithm is  deterministic and thus works against an adaptive adversary.
        %Lastly, we implement our algorithm and compare it with the implementation of~\cite{ChanGS18}
                %which is the state-of-the-algorithm for the problem. Extensive experiments indicate  
%that our algorithm has a significant advantage in running
%time and solution quality for larger values of $k$ and $\eps$ as theory suggests.

	%Our first and main result is an incremental algorithm and we summarize its guarantees in the theorem below.
%>>>>>>> 0064d1f4e8389e3ece0874cd6a87e3445897fd26

\begin{theorem} \label{thm: mainThm1}
		There is a fully-dynamic algorithm for the $k$-center clustering problem, where points are
		taken from a metric space with doubling dimension $\dd$, such that any time the cost of the
		maintained solution is within a factor $(2+\epsilon)$ to the cost of the optimal solution
		and the insertions and deletions of points are supported in $O(2^{O(\kappa)}\logDloglogD \cdot \epsilon^{-1}\ln
		\eps^{-1})$ update time. For any given point, queries about whether this point
		is a center or reporting the cluster this point is assigned to can be answered in $O(1)$ and
		$O(\log \Delta)$, respectively.
	\end{theorem} 
	
We perform an extensive experimental study of a variant of our algorithm, where we replace navigating nets with the closely related notion of \emph{cover trees}, one of the pioneering data-structures for fast nearest-neighbour search~\cite{BKL2006,Kollar2006FASTNN}. We compare our results against the implementation of~\cite{ChanGS18}, which is the state-of-the-algorithm for the problem. Our findings indicate that our algorithm has a significant advantage in running time and solution quality for larger values of $k$ and $\eps$, as suggested by our theoretical results. 

%\paragraph*{Remark.}  
	
   % We show that by slightly increasing the approximation guarantee one can also support the deletion of the points, thus leading to a fully-dynamic algorithm. 
   
	%	\begin{theorem} \label{thm: mainThm2}
	%	There is fully-dynamic algorithm for the $k$-center clustering problem, where points are
	%	taken from a metric space with doubling dimension $\kappa$, such that at any time the cost
	%	of the maintained solution is within a factor of $(2+\epsilon)$ to the cost of the optimal
		%solution and the insertions and deletions of points are supported in $ \normalfont
	%	O(2^{O(\kappa)} \cdot \epsilon^{-1} \ln \eps^{-1} \cdot \poly(\log \Delta))$ update time, where $\Delta$ is the aspect ratio of the metric. For any given point queries about whether this point is a center or reporting the cluster this point is assigned can be answered in $O(1)$ and $O(\log \Delta)$, respectively. 
	% \end{theorem} 

	\paragraph*{Related work.}
	For an in-depth overview of clustering and its wide applicability we refer the reader to two
	excellent surveys~\cite{Schaeffer07,HansenJ97}. Here we briefly discuss closely related variants
	of the $k$-center clustering problem such as the kinetic and the streaming model. In the kinetic
	setting, the goal is to efficiently maintain a clustering under the continuous motion of the
	data points. Gao et al.~\cite{GaoGN06} showed an algorithm that achieves an $8$-approximation factor.
	Their result was subsequently improved to a $(4+ \epsilon)$ guarantee by Friedler and Mount~\cite{FriedlerM10}. In
	the streaming setting Cohen-Addad et al.~\cite{Cohen-AddadSS16} designed a $(6+\epsilon)$-approximation algorithm with an expected update time of $O(k^{2} \cdot \epsilon^{-1} \cdot \log \Delta)$. However, their algorithm only works in the sliding window model and does not support arbitrary insertions and deletions of points. \cite{Kale19} studied streaming algorithms for a generalization of the $k$-center clustering problem, known as the Matroid Center.
	
Recently and independently of our work, Schmidt and Sohler~\cite{SchmidtS19} gave an $16$-approximate fully-dynamic algorithm for the \emph{hierarchical} $k$-center clustering with $O(\log \Delta \log n)$ and $O(\log^{2} \Delta \log n)$ expected amortized insertion and deletion time, respectively, and $O(\log \Delta + \log n)$ query time, where points come from the discrete space $\{1,\ldots,\Delta \}^{d}$ with $d$ being a constant. This result implies a dynamic algorithm for the $k$-center clustering problem with the same guarantees. In comparison with our result, our algorithm (i) achieves a better and an almost tight approximation, (ii) is deterministic and maintains comparable running time guarantees, and (iii) applies to any metric with bounded doubling dimension. For variants of facility location and $k$-means clustering, Cohen-Addad et al.~\cite{Cohen-AddadHPSS19} obtained fully dynamic algorithms with non-trivial running time and approximation guarantees for general metric spaces.  	
	
	There has been growing interest in designing provably dynamic algorithms for graph clustering problems with different objectives. Two recent examples include works on dynamically maintaining expander decompositions~\cite{SaranurakW19} and low-diameter decompositions~\cite{ForsterG19}. For applications of such algorithms we refer the reader to these papers and the references therein.

	\paragraph*{Technical overview.} In the static setting, a well-known approach for designing
	approximation algorithms for the $k$-center clustering problem is exploiting the notion of
	$r$-nets. Given a metric space $(M,d)$, and an integer parameter $r \geq 0$, an \emph{$r$-net} $Y_r$ is a set of points, referred to as \emph{centers}, satisfying (a) the \emph{covering} property, i.e., for every point $x \in M$ there exists a point $y \in Y_r$ within distance at most $2 \cdot (1+\epsilon)^{r}$ and (b) the \emph{separating} property, i.e.,  all distinct points $y,y' \in Y_r$ are at distance strictly larger than $2 \cdot (1+\epsilon)^{r}$. Restricting the set of possible radii to powers of $(1+\epsilon)$ in $(M,d)$ allows us to consider only $O(\epsilon^{-1}\cdot \log \Delta)$ different $r$-nets, where $\Delta$ is the \emph{aspect ratio}, defined as the ratio between the maximum and the minimum pair-wise distance in $(M,d)$. The union over all such $r$-nets naturally defines a hierarchy $\Pi$. It can be shown that the smallest $r$ in $\Pi$ such that the size of the $r$-net $Y_r$ is at most $k$ yields a feasible $k$-center clustering whose cost is within $(2+\epsilon)$ to the optimal one~(see e.g.,~\cite{ChanGS18}).
	
	A natural attempt to extend the above static algorithm to the incremental setting is to maintain
	the hierarchy $\Pi$ under insertions of points. In fact, Chan et al.~\cite{ChanGS18} follow this
	idea to obtain a simple incremental algorithm that has a linear dependency on the number of
	centers $k$. We show how to remove this dependency in metrics with bounded doubling dimension
	and maintain the hierarchy under deletion of points. Concretely, our algorithm employs
	\emph{navigating nets}, which can be thought of as a union over slightly modified $r$-nets with
	slightly larger constants in the cover and packing properties. Navigating nets were introduced
	by Krauthgamer and Lee~\cite{KrauthgamerL04} to build an efficient data-structure for the  nearest-neighbor search
	problem. We observe that their data-structure can be slightly extended to a dynamic algorithm
	for the $k$-center clustering problem that achieves an $8$-approximation with similar update time
	guarantees to those in~\cite{KrauthgamerL04}. Following the work of McCutchen and Khuller~\cite{MK08}, we 
	maintain a carefully defined collection of navigating nets, which in turn allow us to bring down the approximation 
	factor to $(2+\epsilon)$ while increasing the running time by a factor of $O(\epsilon^{-1} \ln \eps^{-1})$.
	
Similar hierarchical structures have been recently employed for solving the dynamic sum-of-radii clustering problem~\cite{HenzingerLM17} and the dynamic facility location problem~\cite{GoranciHL18}. In comparison to our result that achieves a $(2+\epsilon)$-approximation, the first work proves an approximation factor that has an exponential dependency on the doubling dimension while the second one achieves a very large constant. Moreover, while our data-structure supports arbitrary insertions of points, both works support updates only to a \emph{specific} subset of points in the metric space.

	\section{Preliminaries}\label{sec:prelim}
	In the $k$-center clustering problem, we are given a set $\ms$ of points equipped with some metric $d$ and an
	integer parameter $k > 0$. The goal is to find a set $C = \{c_1, \dots, c_k \}$ of $k$
	points (centers) so as to minimize the quantity $\phi(C) = \max_{x \in S} d(x,C)$, where $d(x,C) = \min_{c \in C}  d(x,c)$. 
	Let $\opt$ denote the cost of the optimal solution.

	In the $\emph{dynamic}$ version of this problem, the set $\ms$ evolves over
	time and queries can be asked. Concretely, at each timestep $t$, either a new point is added to
	$\ms$, removed from $\ms$ or one of the following queries is made for any given point $x \in M$: (i) decide whether $x$ is a center in the current solution, and (ii) find the center $c$ to which $x$ is assigned to. The goal is to maintain the set of centers $C$ after each client update so as to maintain a small factor approximation to the optimal solution. 

	Let $\dmin$ and $\dmax$ be lower and upper bounds on the minimum and the maximum distance between any
	two points that are ever inserted. 	
	For each $x \in \ms$ and radius $r$, 
	let $\ball{x}{r}$ be the set of all points in $M$ that are within distance $r$ from $x$, i.e., $\ball{x}{r} := \{ y \in \ms \mid d(x,y) \leq r \}$.

	The metric spaces that we consider throughout satisfy the following property. 

	\begin{definition}[Doubling Dimension]
		The doubling dimension of a metric space $(M,d)$ is said to be bounded
		by $\kappa$ if any ball $\ball{x}{r}$ in $(M,d)$ can be covered by $2^\kappa$
		balls of radius $r/2$.
	\end{definition}
\section{Fully dynamic \texorpdfstring{$k$}{k}-center clustering using navigating nets}\label{sec:navnetskcenter}
	In this section, we present a fully-dynamic algorithm for the $k$-center clustering problem that
	achieves a $(2+\eps)$-approximation with a running time not depending on the number of clusters
	$k$. Our construction is based on navigating
	nets of Krauthgamer and Lee~\cite{KrauthgamerL04} and a scaling technique of McCutchen and Khuller~\cite{MK08}. %which was previously employed for improving the doubling algorithm due to Charikar et al~\cite{CharikarCFM04}. 
	
	We start by reviewing some notation from~\cite{KrauthgamerL04}.

	\paragraph*{$r$-nets and navigating nets.}
	Let $(M,d)$ be a metric space. For a given parameter $r>0$, a subset $Y \subseteq M$ is an $\emph{r-net}$ of $M$ if the following properties hold: 
	\begin{enumerate}
	\itemsep0em
		\item (separating) For every $x,y \in Y$ we have that $d(x,y) \geq r$ and, 
		\item (covering) $M \subseteq \bigcup_{y \in Y} B(y,r)$. 
	\end{enumerate}

	Let $\alpha > 1$ be a constant and let $\Gamma := \{\alpha^i : i \in \Z_{+} \}$ be a set of \emph{scales}. 
	%very value $r \in \Gamma$ a \emph{scale}.
	Let $Y_r:= M$ for all $r \leq \dmin$, and for all $r \in \Gamma$, define $Y_r$ to be an $r$-net of
	$Y_{r/\alpha}$. A \emph{navigating net} $\Pi$ is defined as the union of all $Y_r$ for all $r \in
	\Gamma$. We refer to the elements in $Y_r$ as \emph{centers}.

	Note that for every scale $r > \dmax$ the set $Y_r$ contains only one element due to the
	separating property. 
	A navigating net $\Pi$ keeps track of (i) the smallest scale $r_{\max}$ 
	defined by $r_{\max} = \min\{r \in \Gamma \mid \forall r' \geq r, |Y_{r'}| =1 \}$, and (ii) the largest scale $r_{\min}$ defined by $r_{\min} = \max\{r \in \Gamma \mid  \forall r' \leq r, Y_{r'} = M\}$.
	All scales $r \in \Gamma$ such that $ r \in [r_{\min}, r_{\max}]$ are referred to as \emph{nontrivial} scales.
	%Due to~\cite[Lemma 2.3]{KrauthgamerL04} a navigating net has only $O(\log \Delta)$ nontrivial
	%scales when $\alpha = 2$. Moreover, in~\cite{KrauthgamerL04} it is shown that a navigating net can be maintained
	%in time $O(2^{O(\kappa)} \cdot  \polyD)$ per point insertion and deletion. It is well known how
	%to achieve an 8-approximation for the $k$-center problem using Navigating Nets~\cite{dasguptalecture}.

	\subsection{Navigating nets with differing base distances}\label{subsec:4eps:static}
	In what follows, we describe how to obtain a $(2+\eps)$-approximation for the
	$k$-center clustering problem by maintaining navigating nets in parallel. 
	This technique was originally introduced by McCutchen and Khuller~\cite{MK08} for improving the
	approximation ratio of the incremental doubling algorithm for the $k$-center problem due
	to Charikar et al.~\cite{CharikarCFM04}.

	The key idea behind the construction is that instead of maintaining one navigating net, we
	maintain $m$ navigating nets with differing base distances. The navigating nets differ
	\emph{only}	in the corresponding set $\Gamma$ which is used to define them.
	More concretely, for each integer $1 \leq p \leq m$, let $\hierarchy^p = \{ \alpha^{i+(p/m)-1} \mid i \in \Z_{+} \}$.
	
	%Again, an element $r \in \hierarchy_p$ is called \emph{scale}.
	Let $Y^p_r:= M$ for all $r \leq \dmin$ and for all $r \in \Gamma^p$, let $Y^p_r$ be an $r$-net of
	$Y^p_{r/\alpha}$. A \emph{navigating net} $\Pi^p$ is defined as the union
	over all $Y^p_r$ for $r \in
	\Gamma^p$. Similarly, we maintain $\rpmax$ and $\rpmin$, such that $\rpmax= \min\{r \in \hierarchy^p
		\mid \forall r' \geq r, |Y^p_{r'}| =1 \}$ and $\rpmin = \max\{r \in \Gamma^p \mid  \forall
		r' \leq r, Y^p_{r'} = M\}$, respectively. By definition of $\hierarchy^p$,
		there is an $\alpha^{j/m-1}$-net for all positive integers $j$.

	We next show how to maintain a $k$-center solution for the set of points $M$ using the family of navigating nets $\{\Gamma^{p}\}_{p=1}^{m}$. For each navigating net $1 \leq p \leq m$,
	define $i^* = i+(p/m)-1$ to be the index such that the $\alpha^{i^*}$-net
	$Y_{\alpha^{i^*}}^p$ has at most $k$ centers and $Y_{\alpha^{i^*-1}}^p$ has
	more than $k$ centers.  
	Define $\cost_p = \frac{\alpha}{\alpha-1} \alpha^{i^*}$ for all $1 \leq p
	\leq m$. 
	We compare the costs of all navigating nets and pick the navigating net
	$p^*$ with minimal cost $p^* = \arg \min_{1 \leq p \leq m} \cost_p$. The set of centers $Y_{\alpha^{i^*}}^{p^*}$ is the output $k$-center solution.
	
	The next lemma proves that every point $x \in \ms$ is within a distance
	$\cost_p=\frac{\alpha}{\alpha-1} \alpha^{i^*}$ of
	a center in $Y_{\alpha^{i^*}}^{p}$.

	\begin{lemma}\label{lem:geom:upperb}
		For $1 \leq p \leq m$ and $x \in \ms$ there is a center $c \in Y_{\alpha^{i^*}}^p$ such that
		$d(x,c)\leq \cost_p$.
	\end{lemma}
	\begin{proof}
		By construction, the set $Y_{\alpha^{i^*}}^p$ is an $\alpha^{i^*}$-net of
		$Y_{\alpha^{i^*-1}}^p$ and all elements of $Y_{\alpha^{i^*-1}}^p$ are within distance 
		$\alpha^{i^*}$ to a center in $Y_{\alpha^{i^*}}^p$. Similarly, the
		elements of $Y_{\alpha^{i^*-2}}^p$ are
		within distance $\alpha^{i^*}+\alpha^{i^*-1}$ to a center in
		$Y_{\alpha^{i^*}}^p$ and so on. Note that the set
		$Y_{r^p_{\min}}^p$	contains all points currently in $M$ and thus the distance of every point in $M$ to some center in $Y_{\alpha^{i^*}}^p$ forms a geometric series. 
		
		Formally, let $x \in \ms$ be arbitrary and let $c$
		 be its ancestor in $Y_{\alpha^{i^*}}^p$.  Then the distance between $c$ and $x$ is bounded as follows:

		 \begin{align*} d(x,c) &\leq \alpha^{i^*} + {\alpha^{i^*-1}} + \alpha^{i^*-2} +
			 \cdots\\ &\leq \alpha^{i^*} \sum_{i=0}^{\infty}  \left({1 \over \alpha}\right)^i \\&=
			 \alpha^{i^*}\frac{\alpha}{\alpha-1} = \cost_p. \qedhere
 \end{align*}
	\end{proof}

	The above lemma shows an upper bound for the output $k$-center solution $Y_{\alpha^{i^*}}^{p^*}$, i.e.,
	$\phi(Y_{\alpha^{i^*}}^{p^*}) \leq
	\cost_{p^*}$. The next lemma proves that $\cost_{p^*}$ has the desired approximation guarantee, i.e., $\cost_{p^*} \leq (2+\eps) \OPT$. 

	\begin{lemma}\label{2eps:correctness}
		If $\alpha = O(\eps^{-1})$ and $m = O(\eps^{-1} \ln \eps^{-1})$ then
		$\cost_{p^*} \leq (2+\eps) \OPT$.
	\end{lemma}
	\begin{proof}
		We set $p^* \gets \arg \min_{1 \leq p \leq m} \cost_p$, where $\cost_{p^*} =
		\frac{\alpha}{\alpha -1} \alpha^{i^*}=
		\frac{\alpha}{\alpha -1}\alpha^{j/m -1}$ for some $j\in \Z$. 
		For comparison, consider level $\hat{\alpha} = \alpha^{(j-1)/m-1}$ and the corresponding
		$\hat{\alpha}$-net $Y^{\hat{p}}_{\hat{\alpha}}$. 
		Note that we returned $Y_{\alpha^{i^*}}^{p^*}$ 
		instead of $Y_{\hat{\alpha}}^{\hat{p}}$ as a solution even though $\alpha_{i^*} > \hat{\alpha}$.
		Consequently, $|Y_{\hat{\alpha}}^{\hat{p}}| > k \geq  |Y_{\alpha_{i^*}}^{p^*}|$.
		Because $|Y^{\hat{p}}_{\hat{\alpha}}|
		> k$, at least two points $c_1,c_2 \in Y^{\hat{p}}_{\hat{\alpha}}$ are assigned to the
		same center $c^*$ in the optimal solution. By the separation property we get that
		$d(c_1,c_2) \geq \hat{\alpha}$.
		Using the triangle inequality we obtain
		$$2\OPT \geq d(c_1,c^*) + d(c^*,c_2)  \geq d(c_1,c_2) \geq \alpha^{(j-1)/m-1}$$
		and thus $\OPT \geq \alpha^{(j-1)/m-1}/2$.
		To obtain the desired approximation we compare our result with
		$\cost_{p^*}$:
		\begin{align*}
			{\cost_{p^*} \over \OPT} & \leq { \frac{\alpha}{\alpha -1}\alpha^{j/m -1} \over \alpha^{(j-1)/m-1}/2 } \\  & =  { 2\alpha^{j/m+1} \over (\alpha-1) \cdot 
				\alpha^{(j-1)/m} } \\ &= {2\alpha^{(j-1)/m} \cdot \alpha^{1/m+1} \over (\alpha-1) \cdot
				\alpha^{(j-1)/m} } \\ &= {2 \alpha \over (\alpha -1)}\sqrt[m]{\alpha} .
		\end{align*}
		It remains to show that
		$2{\alpha \over (\alpha -1)}\sqrt[m]{\alpha} \leq 2(1+\eps) \cdot (1+\eps)$. 
		Set $\alpha = {2/\eps}$. Clearly, $\alpha = O(\eps^{-1})$ and
		${\alpha \over \alpha-1} = 1 + {\eps \over 2-\eps} \leq 1 +\eps$ because $0 < \eps \leq 1$.
		Moreover note that
			$\alpha^{1/m} \leq (1+\eps) \text{ iff } 1/m \log_{1+\eps} \alpha \leq 1$. The latter holds 
		for  any $m \geq \eps^{-1} \ln 2 + \eps^{-1}
		\ln \eps^{-1}$, which in turn implies that $m = O(\eps^{-1} \ln \eps^{-1})$.\qedhere
	\end{proof}

	\subsection{Fully dynamic \texorpdfstring{$k$}{k}-center clustering}
	In this section, we present the details of the data structure presented in
	Section~\ref{subsec:4eps:static}. 

	\paragraph*{Data structure.} Our data-structure needs to (1) maintain $m$ navigating
	nets and (2) answer queries about our current solution to the given
	$k$-center clustering problem. 

	For (1) we use the data structure described
	in~\cite{KrauthgamerL04}: Let $1 \leq p \leq m$ and $\alpha^i \in \Gamma^p$: For the navigating
	net $\Pi^p$ we do not store the sets $Y^p_{\alpha^i}$ explicitly.
	Instead, for every nontrivial scale $\alpha^i \in \Gamma^p$ and every $x \in Y^p_{\alpha^i}$ we store the
	\emph{navigation list} $L^p_{x,\alpha^i}$ which contains nearby points to $x$ in the
	$\alpha^{i-1}$-net $Y^p_{\alpha^{i-1}}$, i.e., $L^p_{x,\alpha^i} = \{z \in
	Y^p_{\alpha^{i-1}} : d(z,x) \leq \psi \cdot \alpha^i \}$ where
	$\psi \geq 4$. Additionally, for each $x \in M$ and each $1 \leq p \leq m$, we store the largest
	scale $\beta \in \Gamma^p$ such that 
	$L^p_{x,\beta} = \{x\}$ but we do not store any navigation list $L^p_{x,\alpha}$ where $\alpha \in \Gamma^p$ and
$\alpha < \beta$.
	
	%That is each point $x \in M$ stores the navigation lists for all nontrivial scales.

	For (2), we also maintain the reverse information. Specifically, for every $x$ in $M$
	and nontrivial scale $\alpha^i$ we maintain $M_{x,\alpha^i}^p$ which
	contains all the points in the $\alpha^{i+1}$-net
	$Y^p_{\alpha^{i+1}}$ whose navigation list contains $x$, i.e.,
	$M^p_{x,\alpha^i} = \{ y \in Y^p_{\alpha^{i+1} }: x
	\in L^p_{y,{\alpha^{i+1}}} \}$. We maintain each $M^p_{x,\alpha^i}$ in a min-heap data structure, where each
	element $y \in M^p_{x,\alpha^i}$ is stored with the distance $d(x,y)$.
	It is well known that constructing such a min-heap takes $O(|M^p_{x,\alpha^i}|)$ time and the
	insert and delete operations can be supported in logarithmic time in the size of $M^p_{x,\alpha^i}$. 
	Let $y$ be the closest point to $x$ in $M^p_{x,\alpha^i}$.
	The min-heap allows us to extract $y$ in $O(1)$ time. 
	Note that due to the covering property the closest point to 
	$y$ is also the closest point to $x$ in $Y^p_{\alpha^{i+1}}$.

	Additionally, we maintain a counter $c_{\alpha^i}^p = |Y^p_{\alpha^i}|$
	for each scale $\alpha^i \in \Gamma^{p}$ and navigating net $1 \leq p \leq m$.
	Also, for each navigating net $1 \leq p \leq m$, we maintain the largest scale $\alpha^{i^*}$ such that 
	$c^p_{\alpha^{i^*}} \leq k$ and $c^p_{\alpha^{i^*+1}} > k$. 
	We store $\cost_{p^*} = \min_{1 \leq p \leq m}
	\frac{\alpha}{\alpha-1}\alpha^{i^*}$ and $p^* = \arg \min_{1 \leq p \leq 
	m} \cost_p$. 
	%Storing the counters $c_{\alpha^i}^p$ costs us $O(|M|)$  space for every navigating
	%net.

	\paragraph*{Preprocessing.}
	Consider the construction of a single navigating net $\Pi^p$. We start by 
	inserting the $|M|$ points using the routine described in~\cite{KrauthgamerL04}[Chapter~2.5] whose running time is $O(2^{O(\dd)}  \logDloglogD)$.
	Additionally we construct the lists $M^p_{x,\alpha^i}$ for every $1 \leq p \leq m$, $x \in M$
	and scale $\alpha^i$. We do this during the insert operation which takes care of the lists $L^p_{x,\alpha^i}$.
	Due to Lemma~2.2 in~\cite{KrauthgamerL04} every navigation list has size $O(2^{O(\dd)})$ and due
	to Lemma~2.3 in~\cite{KrauthgamerL04} every navigating net has only $\log \Delta$ nontrivial scales.
	Consequently, the sum of all navigation lists in a navigating net $\Pi^p$ is of size $\sum_{x,\alpha^i}
	|L^p_{x,\alpha^i}| = O(|M| 2^{O(\dd)} \log \Delta)$. Notice that
	$\sum_{x,\alpha^i} |L^p_{x,\alpha^i}| = \sum_{x,\alpha^i} |M^p_{x,\alpha^i}|
	$ because the sets $M^p_{x,\alpha^i}$ store the reverse
	information of the sets $L^p_{x,\alpha^i}$. Since there are $m=O(\epsilon^{-1} \ln \eps^{-1})$ navigating nets, the latter yields a construction time of $O(|M| 2^{O(\dd)}\logDloglogD \cdot \eps^{-1} \ln \eps^{-1})$.

	\paragraph*{Handling Point Updates and Queries.}
	To handle point insertions and deletions in the $m$ navigating nets, we invoke the
	routines described in~\cite{KrauthgamerL04}[Chapters~2.5-2.6] for all the navigating
	nets.
	We also keep track of the counters $c^p_{\alpha^i}$ and sets
	$M^p_{x,\alpha^i}$ when we handle the insertion and deletions of points in
	the navigating nets. While updating the
	counters $c^p_{\alpha^i}$ we simultaneously keep track of $\alpha^{i^*}$ for
	all navigating nets and maintain $p^*$. 
	
	We next discuss the query operations that our data-structure supports. First, we answer the query whether a given point $x \in M$ is a center by simply checking if the list
	$L^{p^*}_{x,\alpha^{i^*}}$ exists. Second, given a point $x \in M$ we return its corresponding center in $Y_{\alpha^{i^*}}^{p^*}$
	as follows: First we check if $x$ is a center. 
	If not, we consider $L^{p^*}_{x,\beta} = \{x \}$. Note that $\beta = \alpha^i$ for some $i$. 
	Then we repeatedly determine the navigation list $L^{p^*}_{y',\alpha^{i+1}}$ where
	$y'$ is the center in
	$Y^p_{\alpha^{i+1}}$ which
	contains $x$ within radius $\alpha^{i+1}$ using the min-heap
	$M^{p^*}_{x,\alpha^{i+1}}$. Then we set $i=i+1$ until $i = i^{*}-1$. Once we arrive at the list
	$L^{p^*}_{y'',\alpha^{i^{*}-1}}$ we return $y''$ as the center $x$ is assigned to. 
	
	The correctness of the maintained hierarchies follows from the correctness in~\cite{KrauthgamerL04}.
	Due to Lemma~\ref{2eps:correctness} the set
	$Y^{p^*}_{\alpha^{i^{*}}}$ is a feasible solution to the $k$-center problem whose cost is guaranteed to be within $(2+\eps)$ times the optimum cost. 

	We finally analyze the running time of the update and query operations.
	The time for handling a point insertion and a point deletion in a single navigating net is $O(2^{O(\dd)} \logDloglogD)$~(Theorem~2.5 in~\cite{KrauthgamerL04}). Since we maintain $m=O(\eps^{-1} \ln \eps^{-1})$ navigating nets,  the overall time to handle a point insertion or deletion is $O(2^{O(\dd)} \logDloglogD \cdot \eps^{-1} \ln \eps^{-1})$.  It is
	straightforward to see that maintaining the counters $c^p_{\alpha^i}, \alpha^{i^*},p^*,\beta$ and min-heaps
	$M^p_{x,\alpha^i}$ in all navigating nets can also be done in the same time per update. 
	Determining if a point $x \in M$ is a center can be done in $O(1)$. 	Determining the center of a given point $x \in M$ takes $O(\log \Delta)$ time because
	there are $O(\log \Delta)$ nontrivial scales (Lemma~2.3 in~\cite{KrauthgamerL04}) and thus there are $O(\log \Delta)$
	iterations in the lookup algorithm until the scale $\alpha^{i^*}$ is reached. 
	
	Combining the above guarantees yields Theorem~\ref{thm: mainThm1}.

%	\begin{theorem}
%		There is a fully-dynamic algorithm for the $k$-center clustering problem, where points are
%		taken from a metric space with doubling dimension $\dd$, such that any time the cost of the
%		maintained solution is withing a factor $(2+\epsilon)$ to the cost of the optimal solution
%		and the insertions and deletions of points are supported in $O(2^{O(\kappa)}\epsilon^{-1}\ln
%		\eps^{-1} \logDloglogD)$ update time. For any given point, queries about whether this point
%		is a center or reporting the cluster this point is assigned to can be answered in $O(1)$ and
%		$O(\log \Delta)$, respectively.
%	\end{theorem}

	%The above description and analysis of the data-structure prove Theorem~\ref{thm: mainThm2}. 
\section{Empirical Analysis}
	In this section, we present the experimental evaluation for our $k$-center algorithm. 
	We implemented the algorithm described in the previous sections using cover trees~\cite{BKL2006,Kollar2006FASTNN}, which is a fast variant of navigating nets. The cover tree maintains the same invariants as navigating nets, except that for a point at a certain level in the hierarchy, we store \emph{exactly one} nearby point one level up, instead of a set of points that are nearby. \cite{BKL2006} show that all running time guarantees can be maintained for metric spaces with bounded expansion constant. This in turn implies that using a collection of cover trees yields a $(2+\epsilon)$-approximation for the $k$-center clustering problem. The running
	time for an insertion/deletion of a point in a cover tree is in $O(c^6 \ln |\ms|)$ where $c$ is the expansion constant of
	$\ms$~\cite{BKL2006}\footnote{The expansion constant of $\ms$ is defined as the smallest value $c
		\geq 2$ such that $|\ball{p}{2r}| \leq c |\ball{p}{r}|$ for all $p \in \ms$ and $r > 0$.}. 
	Our algorithm maintains $O(\epsilon^{-1} \ln \eps^{-1})$ cover trees. To obtain the current centers of a
	cover tree, we traverse the tree top-down and add all distinct points until we have $k$ points. Due to
	the \emph{nesting property} of the cover tree, i.e., every point which appears in some
	level $i$ appears in every lower level $j<i$ in the tree~\cite{BKL2006}
	we are guaranteed to add all nodes of the desired level $Y^p_{i^*}$ described in
	Section~\ref{subsec:4eps:static}.
	From now on we call the described algorithm $\ACov$.\footnote{Source code and data sets:~\url{http://bit.ly/2S4WvJL}}
	
	We compare our algorithm against the algorithm 
	of Chan et al.~\cite{ChanGS18} which is the state-of-the-art approach for the fully dynamic $k$-center
	problem in practice.

	\paragraph{The algorithm of Chan et al.~\cite{ChanGS18}.}
	To gain some intuition into the state-of-the art algorithm in practice, we give a brief
	summary of the algorithm described in Chan et al.~\cite{ChanGS18}:
	The algorithm maintains a clustering %$C^r_1, \dots, C^r_{k'}$ and a set $U^r$ for
	for each $r \in \Gamma:=\{(1+\eps)^i : \dmin \leq (1+\eps)^i \leq \dmax, i \in
	\mathbb{N} \}$.  
	%Each $C^r_i$ has a center $c^r_i$ and contains all points $x \in S$ where
	%$d(x,c^r_i) \leq 2r$.
	%When a point $x$ is inserted into $\ms$, it checks all $c^r_i$ (for $1 \leq
	%i \leq k'$) if $x$ is within distance $2r$ of $c^r_i$ and assigns the
	%first such $c^r_i$ as the new center of $x$. If no such center exists and $k' < k$, $x$ is the center of a new cluster
	%$C^r_{k'+1}$. Otherwise, i.e., if $k' = k$, the algorithm sets $U^r \gets
	%U^r \cup \set{x}$. 
	%When a point $x$ is deleted from $\ms$, the algorithm checks whether $x$ is a center of a cluster
	%$C^r_i$ for $1 \leq i \leq k'$. If it is
	%not a center, it is removed from the corresponding $C^r_i$. Otherwise, let $x =
	%c^r_i$ and $\hat{U^r}:= \bigcup_{j=i}^k C^r_j \cup U^r$.
	%For each $j \in [k-i+1,k ]$ the algorithm picks a random center $c^r_j$ from
	%$\hat{U^r}:$ and creates a cluster
	%$C^r_j$ with $c^r_j$ and all points at distance $\leq 2r$ from $c^r_j$
	%and set $\hat{U^r} \gets
	%\hat{U^r} \setminus C^r_j$.
	%To return a solution, the algorithm returns the clusters $\{C^r_1, \dots C^r_{k'}\}$ of the smallest $r \in
	%\Gamma$ where $U^r = \emptyset$. 
	Their algorithm is a $(2+\eps)$-approximation of the optimal
	solution and has an average running time of $O(k^2 \cdot \frac{\log(\ar)}{\eps})$ per
	update. Note that the algorithm needs $\dmin$ and $\dmax$ as input and that these values are
	usually not available in practice. In contrast, $\ACov$ does not need
	these parameters. For our empirical analysis we provided these special
	parameters to the algorithm of Chan et al.~\cite{ChanGS18}. %for the %data sets we used, 
	For arbitrary instances one would initialize $\dmin,\dmax$ with the 
	minimum/maximum value for the type double respectively 
	to guarantee the correctness of their algorithm. From now on, we call their
	algorithm $\AChan$.

	%To initialize the algorithm of Chan et al. we used the maximum- and minimum
	%distances for $d_{min}$ and $d_{max}$ in the corresponding datasets.

	\paragraph{Setup.}
	We implemented the cover tree in C++ and compiled it with
	g++-7.4.0. We executed all of our experiments on a Linux machine running on
	an AMD Opteron Processor 6174 with 2.2GHz and 256GB of RAM. 
	In our experiments we evaluate $\AChan$
	and $\ACov$ with the following pairwise combinations of $\epsilon \in \{0.1, 0.5, 1,4\}$
	and $k \in \{20, 50, 100, 200\}$. In total, we perform 10 different runs for each test instance
	and compute the arithmetic mean of the solution improvement and speedup on this instance. When further averaging over
	multiple instances, we use the geometric mean in order to give every instance a comparable
	influence on the final score. 
	To measure the solution quality of an algorithm at any timepoint $i$ we query for the current set of
	centers $C_i$. We do not directly compute the objective function value $\phi(C_i)$, since this is an expensive
	operation and it is not usually needed in practice.
	After the termination of the two algorithms we compute the	objective function of the $k$-center
	solution $\phi(C_i)$ in order to compare the solutions of the two competing algorithms $\ACov$
	and $\AChan$. Hence, the running times of both algorithms include the time to
	perform the point insertions/deletions and the queries (obtaining the centers of the solution), but not computing
	the objective function.

	\paragraph{Instances and Update Sequences.}
	To compare the performance of the two algorithms, we use the instances of Chan et al.~\cite{ChanGS18} with euclidean distance and add an additional random instance.
	%\vspace{-1em}
	\begin{itemize}
		\itemsep0.01mm
		\item \emph{Twitter.} The twitter data set~\cite{chan_github} is introduced
			in~\cite{ChanGS18} and consists of
			 21 million geotagged tweets.  Our
			experiments consider only the first 200k tweets without duplicates. 
		\item \emph{Flickr.} The Yahoo Flickr Creative Commons 100 Million (YFCC100m)
			dataset~\cite{yfcc100m} contains the metadata of 100 million pictures posted on Flickr. Unfortunately, we were not
			able to obtain the full dataset but used a search engine to build a subset of
			the dataset~\cite{flickr_searchengine}. This subset entails 800k points with longitude and latitude. 
		\item \emph{Random.} This dataset consists of 2 million points created as follows: 
			First, we sampled $100$ points $(x,y)$ uniformly at random for $-1 \leq x,y \leq 1$.
			Then, for each such point $(x,y)$, we sampled another $20000$ points using a normal
			distribution with $(x,y)$ as mean and a variance of $0.001$ respectively. %Notice that due
			%to this setup, it is very likely that the optimal solution is a set $C$,
			%where $|C| = 100$ and $\phi(C) \approx 0.001$.
	\end{itemize}

	We use the following update sequences on the data sets inserting at most 200k points:
	%\vspace{-1em}
	\begin{itemize}
		\itemsep0.01mm
		\item \emph{Sliding Window.} In the sliding window query, a point is inserted at some point
			in time $t$ and will be removed at time $t+W$ where $W$ is the window size. We chose a
			sliding window of size 60k following the implementation of Chan et al. During the
			update sequence we perform a query every 2000 insertions. Therefore, we
			perform 100 queries in total. 

		\item \emph{Random Insertions/Deletions.} We further distinguish between three concrete types
			of update sequences with 30\% deletions, 10\% deletions and 5\% deletions. Points are inserted
			uniformly at random and deleted uniformly at random from the set of points already inserted. 
			The chance to perform a query is 0.05\%. 
			The chance to insert a point at any given timestep is given by
			1 - the respective deletion percentage above - 0.0005.
	\end{itemize}

	\paragraph{Results and Interpretation.}
	We now evaluate the performance of $\ACov$ and compare it $\AChan$. 
	In Table~\ref{tab:running_epsfixed} we
	present the geometric mean speedup of $\ACov$ over $\AChan$. Here, both algorithm use the same
	parameter $\epsilon$ and have the same number of centers $k$. First of all, note that the empirical results
	\emph{reinforce} the
	theoretical results: The larger $k$ and $\eps$ are in our experiments, the larger the speedups of $\ACov$ become when compared to the algorithm $\AChan$.
	The running time of our algorithm $\ACov$ does not depend on $k$
	whereas in contrast each updates of $\AChan$ depends \emph{quadratically} on $k$ on average. 
	Moreover, speedups improve for larger values of $\eps$ since the running time of
	$\ACov$ has a multiplicative factor of $O(\eps^{-1} \ln \eps^{-1})$ and
	$\AChan$'s running time includes a better factor $O(\eps^{-1})$.
	For example, when $k$ is as large as $200$,  $\ACov$ is faster than $\AChan$ for all values of
	$\epsilon$. In contrast, when $\eps=1$, $\ACov$ has better speedups than $\AChan$ already for small values
	of $k$ like $k=50$.  
	When $k = 20$, $\AChan$ is faster than $\ACov$. 

	We proceed to compare the solution quality when both algorithm use the \emph{same} parameter $\epsilon$ and also use the same number of centers $k$.
	In Table~\ref{tab:solv_epsfixed} we present the geometric mean solution
	improvement of $\ACov$ over $\AChan$ for this case. $\ACov$ gives better solutions for
	\emph{all}
	instances as soon as $\epsilon \geq 0.5$. Generally speaking, the larger $\epsilon$ gets, the larger is our improvement in
	the	solution: For $\epsilon=0.5$ our algorithm gives 10-12\%
	better solutions. Setting $\epsilon = 1$ we already obtain 12-36\% better solutions and finally,
	when setting $\epsilon = 4$ we obtain 7-114\% better solutions.
	For $\epsilon=0.1$ our solutions are about 3-4\% worse than the
	solutions of $\AChan$. We conclude that our algorithm has a significant advantage in running
	time \emph{and} solution quality for slightly larger values of $k$ and $\eps$.
		\begin{table}[t!]
		\centering
		\caption{Top: Geometric mean speedup of $\ACov$ over $\AChan$. Bottom: Geometric mean improvement in solution quality of $\ACov$ over $\AChan$. For every entry both algorithms use the same $\eps$ and $k$. Higher is better.}\label{tab:running_epsfixed}
\label{tab:solv_epsfixed}
		\begin{tabular}{lrrrr}
			\toprule
			$\epsilon$ &       0.1 &       0.5 &        1.0 &        4.0 \\
			$k$   &           &           &            &            \\
			\midrule
			20  &  0.02 &  0.14 &   0.32 &   0.72 \\
			50  &  0.10 &  0.59 &   \textbf{1.34} &   \textbf{3.05} \\
			100 &  0.33 &  \textbf{2.01} &   \textbf{4.45} &  \textbf{10.32} \\
                        200 &  \textbf{1.15} &  \textbf{7.66} &  \textbf{17.74} &  \textbf{39.60} \\
			%\bottomrule
\midrule
			20  &  0.97 &  \textbf{1.12} &  \textbf{1.27} &  \textbf{1.07} \\
			50  &  0.97 &  \textbf{1.10} &  \textbf{1.36} &  \textbf{1.46} \\
			100 &  0.96 &  \textbf{1.12} &  \textbf{1.12} &  \textbf{2.14} \\
			200 &  0.96 &  \textbf{1.12} &  \textbf{1.19} &  \textbf{1.28} \\
			\bottomrule

		\end{tabular}
		%\caption{Geometric mean improvement in solution quality of $\ACov$ over $\AChan$. For every entry both algorithms use the same $\eps$ and $k$. Higher is better.}\label{tab:solv_epsfixed}
	%\end{table}
	%\begin{table}[t!]
		\centering
\caption{Top: Geometric mean speedup over $\AChan$ when fixing $\epsilon = 1$ for $\ACov$. Bottom: Geometric mean improvement in solution quality when fixing $\epsilon = 1$ for our algorithm $\ACov$. Higher is better.}\label{tab:solv_eps1}\label{tab:running_eps1}
\vspace{0.3cm}
		\begin{tabular}{lrrrr}
			\toprule
			$\epsilon$ &       0.1 &       0.5 &       1.0 &       4.0 \\
			$k$   &           &           &           &           \\
			\midrule
			20  &\textbf{    2.48} &   0.55 &   0.32 &  0.14 \\
			50  &\textbf{   10.01} &  \textbf{ 2.27} &\textbf{   1.34} &  0.62 \\ 100 &\textbf{   32.17} &\textbf{   7.51} &\textbf{   4.45} &\textbf{ 2.08} \\ 200 &\textbf{  130.01} &\textbf{  29.60} &\textbf{  17.74} &\textbf{
				8.35} \\

	\midrule
			20  &  0.91 &  \textbf{1.08} &  \textbf{1.27} &  \textbf{1.18} \\
			50  &  0.90 &  \textbf{1.05} &  \textbf{1.36} &  \textbf{1.72} \\
			100 &  0.89 &  \textbf{1.06} &  \textbf{1.12} &  \textbf{2.52} \\
			200 &  0.88 &  \textbf{1.06} &  \textbf{1.19} &  \textbf{1.54} \\

			\bottomrule
		\end{tabular}
		
	%\end{table}
	%\begin{table}[t!]
		\centering
		\caption{Top: Geometric mean speedup over $\AChan$ when fixing $\epsilon = 4$ for $\ACov$. Bottom: Geometric mean improvement in solution quality when fixing $\epsilon = 4$ for our algorithm $\ACov$. Higher is better.}\label{tab:running_eps4}
\label{tab:solv_eps4}		
\vspace{0.3cm}		
		\begin{tabular}{lrrrr}
			\toprule
			$\epsilon$ &         0.1 &         0.5 &        1.0 &        4.0 \\
			$k$   &             &             &            &            \\
			\midrule
			20  &\textbf{   12.09} &    2.70 &   1.58 &   0.72 \\
			50  &\textbf{   48.69} &   11.07 &   6.54 &   3.05 \\
			100 &\textbf{  159.11} &   37.17 &  22.03 &  10.32 \\
			200 &\textbf{  616.51} &  140.38 &  84.13 &  39.60 \\
			\midrule
			20  &  0.83 &  0.98 &  1.16 &  1.07 \\
			50  &  0.76 &  0.89 &  1.16 &  1.46 \\
			100 &  0.76 &  0.90 &  0.95 &  2.14 \\
			200 &  0.74 &  0.88 &  0.99 &  1.28 \\
			\bottomrule
		\end{tabular}
	\end{table}

	%. The Table~\ref{tab:running:eps_1
    We now fix the value of $\eps$ in our algorithm to $1$ and $4$ and compare it with $\AChan$ for all values of $\eps$.
	%Due to the fact that we improve the solution of $\ACov$ over $\AChan$ for all $\eps\geq 0.5$ we fixed
	%$\eps$ to $1$ for $\ACov$ and compared it with results of $\AChan$ again for all values of $\eps$. 
	Table~\ref{tab:running_eps1} presents the geometric mean speedup of the results 
	and  the geometric mean improvement in solution quality for the case that we fix $\eps=1$ in our algorithm. Notice that we
	obtain a speedup of at least one order of magnitude when $k\geq 50$ comparing to $\AChan$ with $\eps =
	0.1$ while sacrificing only 9-12\% in solution quality over $\AChan$.  Most significantly, $\ACov$ is faster
	than $\AChan$ with $\eps = 0.5$ and
	$k\geq 50$ while also obtaining \emph{better} solution quality. 
	Similarly, we set $\eps = 4$ for $\ACov$ and compare the results to $\AChan$ for all values of $\epsilon$ again.
	The resulting geometric mean speedups  and the geometric mean
	solution improvement is presented in Table~\ref{tab:solv_eps4}.
	When comparing to $\AChan$ with $\eps=0.1$ we obtain speedups of one order when $k\leq 50$ and
	two orders when $k \geq 100$ while sacrificing at most 26\% of the solution quality.
	%\vfill

\section{Conclusion}
We developed a fully dynamic $(2+\eps)$ approximation algorithm for k-center
clustering with running time independent of $k$, the number of centers. 
Our algorithm maintains multiple hierarchies (so called navigating nets), so that each hierarchy
stores sets of points which evolve over time through deletions and insertions.
Roughly speaking, each of these hierarchies maintains the property that points residing on the same level are at least separated by a specific distance. This allows us to
obtain k-center solutions with an approximation of $(2+\eps)$. Maintaining the navigating nets can be done in 
time independent of $k$.
Lastly, we conducted an extensive evaluation of this algorithm which indicates that our algorithm outperforms the state-of-the-art
algorithms for values of $k$ and $\eps$ suggested by theory. In this case, our algorithm obtains significant speedups and improvements in solution quality.
Important future work includes parallelization of the two algorithms as well as
implementing the streaming algorithms in \cite{SchmidtS19,MK08} and \cite{CharikarCFM04}.   
\bibliography{cluster}
\end{document}